\documentclass[12pt]{iopart}

\usepackage{graphicx}
\usepackage{dcolumn}
\usepackage{bm}
\usepackage{multirow}
\usepackage{url}
\usepackage{amsthm}
\usepackage{cite}
\usepackage{algorithm}
\usepackage{algorithmic}
\newcommand{\trace}{\mathop{\mathrm{Tr}}}
\newtheorem{theorem}{Theorem}
\newtheorem{assumption}{Assumption}
\newtheorem{definition}{Definition}

\begin{document}

\title{Computing the Classical-Quantum channel capacity: experiments on a  Blahut-Arimoto type algorithm and an approximate solution for the binary inputs, two-dimensional outputs channel}

\author{Haobo Li$^1$ and Ning Cai$^2$}

\address{$^{1,2}$ 199 Huanke Road, Pudong New Area, Shanghai, China}
\ead{$^1$ lihb@shanghaitech.edu.cn}
\ead{$^2$ ningcai@shanghaitech.edu.cn}
\vspace{10pt}
\begin{indented}
\item \date{\today}
\end{indented}

\begin{abstract}
In our previous work \cite{myBA}, we presented a Blahut-Arimoto type algorithm for computing the discrete memoryless (DM) classical-quantum channel capacity. And the speed of convergence is analyzed. In this paper, we present numerical experiment to show that the algorithm converged much faster than the theoretical prediction. Also we present an explicit approximate solution for the binary inputs, two dimensional outputs classical-quantum channel, which has an error within $3\times10^{-4}$
\end{abstract}

\section{Introduction}
The classical-quantum channel \cite{Holevo} can be considered as consisting of an input alphabet $\mathcal{X}=\{1,2,\dots,|\mathcal{X}|\}$ and a mapping $x\rightarrow \rho_x$ from the input alphabet to a set of quantum states in a finite dimensional Hilbert space $\mathcal{H}$. The state of a quantum system is given by a density operator $\rho$, which is a positive semi-definite operator with trace equal to one. Let $\mathcal{D}^m$ denote the set of all density operators acting on a Hilbert space $\mathcal{H}$ of dimension $m$.  If the resource emits a letter $x$ with probability $p_x$, the output would be $\rho_x$, so the output would form an ensemble:  $\{p_x:\rho_x\}_{x\in \mathcal{X}}$.

In 1998, Holevo showed \cite{holevo1} that the classical capacity of the classical-quantum channel is the maximization of a quantity called the Holevo information over all input distributions. The Holevo information $\chi$ of an ensemble $\{p_x:\rho_x\}_{x\in \mathcal{X}}$ is defined as
\begin{eqnarray}\label{holevo}
	\chi(\{p_x:\rho_x\}_{x\in \mathcal{X}})=H(\sum_xp_x\rho_x)-\sum_xp_xH(\rho_x),
\end{eqnarray}
where $H(\cdot)$ is the von Neumann entropy which is defined on positive semidefinite matrices:
\begin{eqnarray}\label{von}
	H(\rho)=-\trace (\rho\log \rho).
\end{eqnarray}
Due to the concavity of von Neumann entropy \cite{wilde}, the Holevo information is always non-negative. The Holevo quantity is concave in the input distribution \cite{wilde}, and  so the maximization of equation (\ref{holevo}) over $p$ is a convex optimization problem. However it is not a straightforward convex optimization problem.

For discrete memoryless classical channels, the capacity can be computed efficiently by using an algorithm called Blahut-Arimoto (BA) algorithm \cite{Arimoto}\cite{blahut}\cite{yeung}. In 1998, H. Nagaoka \cite{original} proposed a quantum version of BA algorithm. In his work he considered the quantum-quantum channel and this problem was proved to be NP-complete \cite{npc}. And Nagaoka mentioned an algorithm concerning classical-quantum channel, however, its speed of convergence was not studied there. In 2014, Davide Sutter et al. \cite{Sutter} promoted an algorithm based on duality of convex programing and smoothing techniques \cite{nestrov} with a complexity of $O(\frac{(n\vee m)m^3(\log n)^{1/2}}{\varepsilon})$, where $n\vee m=\max\{n,m\}$.

The structure of this paper is organized as following: In Section \ref{baa}, we summarize the results of our previous work: the problem and the algorithm.
 In Section \ref{secnum} we show the numerical experiments of BA algorithm to see how well this algorithm performs. In Section \ref{secapp}, we propose an approximate solution for a special case, which is the binary inputs, two dimensional outputs channel.

$\bf{Notations}$: The logarithm $\log(\cdot)$ is with basis $2$ unless specially specified. The space of all Hermitian operators of dimension $m$ is denoted by $H^m$. The set of all density matrices of dimension $m$ is denoted by $\mathcal{D}^m=\{\rho\in H^m:\rho\geq 0,\trace \rho=1\}$. Each letter $x\in\mathcal{X}$ is mapped to a density matrix $\rho_x$ so the classical-quantum channel can be represented as a set of density matrices $\{\rho_x\}_{x\in\mathcal{X}}$. The set of all probability distributions of length $n$ is denoted by $\Delta_n=\{p:p_x\geq0,\sum_{x=1}^np_x=1\}$. The von Neumann entropy of a density matrix $\rho$ is denoted by $H(\rho)=-\trace[\rho\log \rho]$. The relative entropy between $p,q\in\Delta_n$, if ${supp}(p)\subset{supp}(q)$,  is denoted by $D(p||q)=\sum_x p_x(\log p_x-\log q_x)$ and $+\infty$ otherwise. The relative entropy between $\rho,\sigma\in \mathcal{D}^m$, if ${supp}(\rho)\subset {supp}(\sigma)$, is denoted by $D(\rho||\sigma)=\trace[\rho(\log \rho-\log\sigma)]$ and $+\infty$ otherwise.

\section{The BA algorithm and its convergence}\label{baa}
We want to compute the capacity of a given classical-quantum channel $\{x\rightarrow \rho_x\}_{x\in\mathcal{X}}$, which can be  written as 
\begin{eqnarray}\label{primal}
&\max\limits_{p}H(\sum_x  p_x\rho_x)-\sum_{x}  p_xH(\rho_x),\\
subject \  to \  \ &  \langle s|p\rangle\leq S;\label{cons}\\ 
& p\in \Delta_n, 
\end{eqnarray}
where $\rho_x\in \mathcal{D}^m$, $|s\rangle\geq 0,S>0$. We denote the maximal value of equation (\ref{primal}) as $C(S)$. In this optimization problem, we are to maximize the Holevo quantity with respect to the input distribution $\{p_x\}_{x\in\mathcal{X}}$.  Practically, the preparation of different signal state $x$ has different cost, which is represented by $|s\rangle$. And we would like to bound the expected cost of the resource within some quantity, which is represented by the inequality constraint in equation (\ref{cons})

In this paper we only present the algorithm which does not concern the inequality constraint (\ref{cons}). (The algorithms, with or without inequality constraint (\ref{cons}), are essentially the same.)
The algorithm we presented at \cite{myBA} can be summarized as following: (we take the natural logarithm for conveniece)

\begin{algorithm}[H]
\caption{Blahut-Arimoto algorithm for discrete memoryless classical-quantum channel}
\label{As1}
\begin{algorithmic}
\vspace{0.6em}

\STATE {set $p^0_x=\frac{1}{|\mathcal{X}|}$, $x\in \mathcal{X}$;} 
\vspace{0.6em}

\REPEAT 
\vspace{0.6em}

\STATE $p^{t+1}_x=\frac{r_x^t}{\sum_xr_x^t}$, where 
\quad$r_x^t=\exp{(\trace{\{\rho_x[\log{(p^t_x\rho_x)-\log{\rho^t}}]\}}}$, $\rho^t=\sum_xp_x^t\rho_x$;
\vspace{0.6em}

\UNTIL{convergence.} 
\end{algorithmic}
\end{algorithm}

And we have the following theorems: (For detailed derivation and proof, please refer to \cite{myBA}.)
\begin{definition}
	Define the value of  Holevo quantity (\ref{holevo}) of  a certain channel after $t $ iterations as 
	\begin{eqnarray}
		\chi(t) = H(\sum_x  p^t_x\rho_x)-\sum_{x}  p^t_xH(\rho^t_x).
	\end{eqnarray}
\end{definition}
\begin{theorem}\label{th1}
	$\chi(t)$ converges to the capacity  of the classical-quantum channel from below as $t\rightarrow\infty$; and the input distribution $p^t$ converges to an optimal distribution $p^*$ as $t\rightarrow \infty$.
\end{theorem}

\begin{theorem}\label{th2}
	To reach $\epsilon$ accuracy to the capacity, the BA algorithm needs an iteration complexity less than $\frac{\log n}{\epsilon}$.
\end{theorem}

If the classical-quantum channel has some extra feature, we have a better convergence performance:
\begin{assumption}\label{as1}
		The channel matrices $\{\rho_x\}_{x\in \mathcal{X}}$ are linearly independent in the complex matrix space, i.e. there doesn't exist a  vector $|c\rangle$ such that
	\begin{eqnarray}
		\sum_x c_x\rho_x=0.
	\end{eqnarray}
\end{assumption}

\begin{theorem}
	Under Assumption \ref{as1}, the optimal solution $p^*$ is unique. And   $p^t$ converges to $p^*$ at a geometric speed, i.e. there exist $N_0$ and $\delta>0$, where $N$ and $\delta$ are independent, such that for any $t>N_0$, we have
	\begin{eqnarray}
		 D(p^*||p^{t})\leq (1-\delta)^{t-N_0}D(p^*||p^{N_0}).
	\end{eqnarray}

\end{theorem}
\section{Numerical experiments on BA algorithm}\label{secnum}
 We will study the relations between iteration complexity and $n,m$ (i.e. the input size and output dimension) when 
the algorithm reaches certain accuracy. Due to we don't know the true capacity of a certain channel, we will use the following theorem to bound the error of the algorithm.
\begin{theorem}
	With the iteration procedure in the BA algorithm \ref{As1}, $\max_{x}\{D(\rho_x||\rho^t)\}$ converges to $F(\lambda)$ from above.
\end{theorem}
\begin{proof}
	Followed from Algorithm \ref{As1}, and Theorem \ref{th1}, with some calculation, we have
	\begin{eqnarray}
		\lim_{t\rightarrow\infty}\frac{p^{t+1}_x}{p^t_x}=\exp[D(\rho_x||\rho^*)-C],
	\end{eqnarray} 
	where $\rho^*=\sum_xp^*_x\rho_x$, $p^*$ is the optimal distribution that $p^t$ converges to, and $C$ is the capacity of the channel. The limit above is $1$ if $p^*_x>0$ and does not exceed $1$ if $p^*_x=0$. So
	\begin{eqnarray}
		D(\rho_x||\rho^*)\leq C
	\end{eqnarray}
	for every $x\in\mathcal{X}$, with equality if $p^*_x>0$. This proves
	\begin{eqnarray}
		\max_{x}\{D(\rho_x||\rho^t)\}\rightarrow C.
	\end{eqnarray}
	And  for any $p^t$ and any optimal distribution $p^*$, we have 
	\begin{eqnarray}
		\max_{x}&[D(\rho_x||\rho^t)]\geq \sum_x p^*_x[D(\rho_x||\rho^t)]\\
		&=\sum_x p^*_xD(\rho_x||\rho^*)+D(\rho^*||\rho^t)\\&=C+D(\rho^*||\rho^t)\geq C.
	\end{eqnarray}
	The  first equality requires some calculation and the second equality follows since $p^*$ is an optimal distribution. And the last inequality follows from the non-negativity of mutual information. This means $\max_{x}\{D(\rho_x||\rho^t)\}$ converges to $C$ from above.
\end{proof}

So our accuracy criterion is: for a given classical-quantum channel, we run the BA algorithm (with no input constraint), until $[\max_{x}\{D(\rho_x||\rho^t)\}-\chi(t)]$ is less than $10^{-k}$, and record the number of iteration. At this time, the accuracy is of order $10^{-(k+1)}$ since $\max_{x}\{D(\rho_x||\rho^t)\}$ and $\chi(t)$ converges to the true capacity from above and below respectively.

We do the following numerical experiments: for given values of input size $n$ , output dimension $m$ and accuracy, we generate $200$ classical-quantum channels randomly, and record the numbers of iterations then calculate the average number of iterations and find the maximum number of iterations in these $200$ experiments. The results are shown in Table \ref{table11}. Note that the accuracy $10^{-k}$ in Table \ref{table11} means we run the BA algorithm until $[\max_{x}\{D(\rho_x||\rho^t)\}-\chi(t)]$ is less than $10^{-k}$, and the error between the true capacity and the computed value is of order $10^{-(k+1)}$.
\renewcommand\arraystretch{1.3}
\begin{table}[h]
	\centering
	\caption{}
	\label{table11}
	\begin{tabular}{|c |c|c|c|c|}
		\hline
 Input  $n$ & Output $m$ & Accuracy &Ave iteration &Max ite ($\times10^{2}$)\\
 \hline
 \multirow{3}*{2}&  \multirow{3}*{2} & $10^{-3}$&7 &0.35  \\
 &  & $10^{-4}$& 25&0.12  \\
 &  & $10^{-5}$&64 & 0.25 \\
 \hline
  \multirow{3}*{2}&  \multirow{3}*{5} & $10^{-3}$& 11& 0.27 \\
 &  & $10^{-4}$&24 & 0.11 \\
 &  & $10^{-5}$&44 & 0.13 \\
 \hline
  \multirow{3}*{2}&  \multirow{3}*{8} & $10^{-3}$&11 & 0.25 \\
 &  & $10^{-4}$&24 & 0.70 \\
 &  & $10^{-5}$&39 & 0.11\\
 \hline
  \multirow{3}*{5}&  \multirow{3}*{2} & $10^{-3}$& 125& 5.0 \\
 &  & $10^{-4}$& 245&  4.5\\
 &  & $10^{-5}$&494 &  23\\
 \hline
  \multirow{3}*{5}&  \multirow{3}*{5} & $10^{-3}$&95 &  4.5\\
 &  & $10^{-4}$&243 &  10\\
 &  & $10^{-5}$&456 & 20 \\
 \hline \multirow{3}*{5}&  \multirow{3}*{8} & $10^{-3}$& 85& 0.17 \\
 &  & $10^{-4}$&212 & 0.60 \\
 &  & $10^{-5}$&414 &20  \\
 \hline
  \multirow{3}*{8}&  \multirow{3}*{2} & $10^{-3}$& 102& 4.1 \\
 &  & $10^{-4}$& 219& 12\\
 &  & $10^{-5}$&353 &  23\\
 \hline
  \multirow{3}*{8}&  \multirow{3}*{5} & $10^{-3}$&114 &  4.4\\
 &  & $10^{-4}$&275 &  13\\
 &  & $10^{-5}$&456 & 27 \\
 \hline \multirow{3}*{8}&  \multirow{3}*{8} & $10^{-3}$& 112& 3.0 \\
 &  & $10^{-4}$&291 & 9.8 \\
 &  & $10^{-5}$&582 &25  \\
 \hline
	\end{tabular}
\end{table}

We can see from the table that the iteration complexity scales good as accuracy increases. Notice that the datas of $n=5$ and $n=8$ are vary similar which means the iteration complexity also scales very good as the input size $n$ increases. We can also see for given input size $n$ and accuracy, the output dimension has vary little influence on iteration complexity, which means the iteration complexity also scales good as the output dimension $m$ increases. Compared with our theoretical analysis of iteration complexity in Theorem \ref{th2}: to reach $\epsilon$ accuracy, we need $\frac{\log n}{\epsilon}$ number of iterations.  numerical experiments show that both the average and maximum number of iterations are far smaller than $\frac{log n}{\epsilon}$ to reach $\epsilon$ accuracy, no matter whether  the output quantum states are linearly independent (cases in $(n,m)=(5,2), (8,2)$) or not . The reason of this phenomenon is that the inequalities we used in the proof of Theorem \ref{th2} are quite loose. So Theorem \ref{th2} only provide a very loose upper bound on iteration complexity. In \cite{Sutter}, only the binary two dimensional case is put into numerical experiments and the iteration complexity is of order $1/\varepsilon$  to reach $\varepsilon$ accuracy. So compared with the results in \cite{Sutter}, our algorithm is much better both theoretically and practically (Please refer to \cite{myBA} for the theoretical analysis).

\section{An approximate solution of $p$ in binary two dimensional case}\label{secapp}

In this section we provide an approximate optimal input distribution for the case of binary inputs, two dimensional outputs channel:
\begin{eqnarray}
	\{p_1:\rho_1;p_2:\rho_2\},\quad p_1+p_2=1,\rho_1,\rho_2\in \mathcal{D}^2.
\end{eqnarray}

\subsection{Use Bloch sphere to get an approximate solution}
Any two-dimensional density matrix can be represented as a point in the Bloch sphere \cite{wilde}, as shown in the following:

{\centerline{\includegraphics[width=8cm]{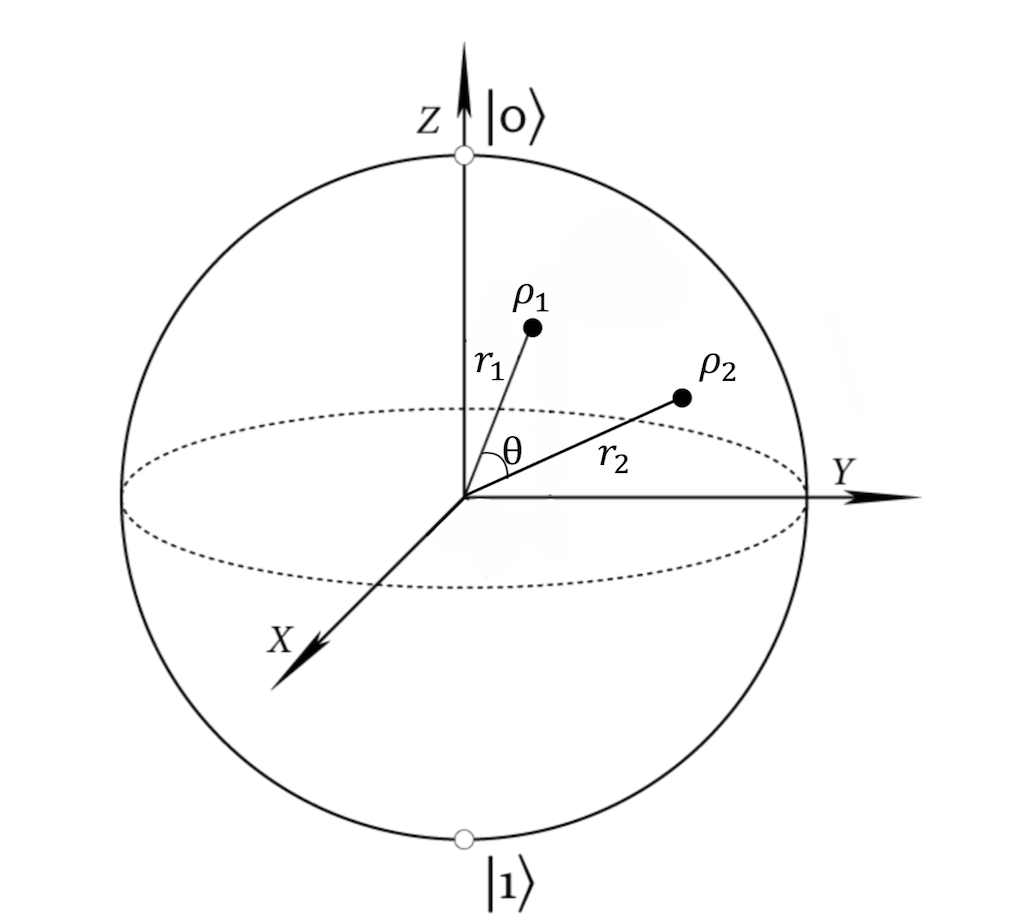}}}
Any density matrix can be represented as a vector in the Bloch sphere starting from the origin. Suppose $\rho_1,\rho_2$ can be represented as $\bm{r_1},\bm{r_2}$ respectively, then the two eigenvalues would be $0.5\pm r_1/2$ and $0.5\pm r_2/2$ respectively. Extend $\bm{r_1}$ we get two intersections on the surface of the Bloch sphere, then these two intersections represents the two eigenvectors of $\rho_1$ (the points on the surface of the sphere represent pure state and the interior points represent mixed states.) A probabilistic combination of $\rho_1,\rho_2$ can be represented as $p_1\rho_1+p_2\rho_2=p_1\bm{r_1}+p_2\bm{r_2}$ (\cite{wilde} Exercise 4.4.13). And any point on the surface of Bloch sphere can be represented as
\begin{eqnarray}
	\cos\frac{\alpha}{2}|0\rangle+\sin\frac{\alpha}{2}e^{i\phi}|1\rangle,
\end{eqnarray} 
where $\alpha$ is the angle to the $Z$ axis and $\phi$ is the angle of the $X$ axis to the projection  of the point on the $X-Y$ plane.

By symmetry, it is obvious that the Holevo quantity is only related to $r_1,r_2,\theta,p_1$, where $\theta$ is the angle between $\bm{r_1},\bm{r_2}$. One interesting result is that the angle $\theta$ has very little influence on $p^*$, where $p^*$ is the optimal distribution that maximizes Holevo quantity. If we know $\lambda_1,\lambda_2$ (the bigger eigenvalues of $\rho_1,\rho_2$ respectively), $\theta$ and $p_1$, then the Holevo quantity can be written as
\begin{eqnarray}\label{chi}
	\chi(\lambda_1,\lambda_2,\theta,p_1)=&S(\frac{1}{2}+||p_1\bm{r_1}+(1-p_1)\bm{r_2}||_2)\\
	&-[p_1S(\frac{1}{2}+r_1)+(1-p_1)S(\frac{1}{2}+r_2)],
\end{eqnarray}
where $S(\cdot)$ is the binary entropy ($S(x)=-(x\log x+(1-x)\log (1-x))$ and $r_i=\lambda_i-\frac{1}{2},i=1,2$.

Using Cosine theorem to calculate $||p_1\bm{r_1}+(1-p_1)\bm{r_2}||_2)$, the gradient of $\chi$ w.r.t. $p_1$ can be calculated directly, denoted as 
\begin{eqnarray}
\nabla_{p_1}\chi(\lambda_1,\lambda_2,\theta,p_1).
\end{eqnarray}

If we can find a $\hat{p}_1$ such that $\nabla_{p_1}\chi(\lambda_1,\lambda_2,\theta,p_1)|_{p_1=\hat{p}_1}=0$, then this $\hat{p}_1$ is the optimal solution (because $\chi(\lambda_1,\lambda_2,\theta,p_1)$ is concave in $p_1$). However, we cannot solve the equation $\nabla_{p_1}\chi(\lambda_1,\lambda_2,\theta,p_1)=0$ w.r.t. $p_1$ when $\theta\ne 0$. Now that $\theta$ has little influence on the optimal distribution $p^*$,  let $\theta=0$ (this is actually the classical case), and let
\begin{eqnarray}
	\nabla_{p_1}\chi(\lambda_1,\lambda_2,\theta=0,p_1)=0,
\end{eqnarray}
the above equation is easy to solve and we  get a solution $\hat{p}_1$:

\begin{eqnarray}\label{app}
	\hat{p}_1=\frac{\frac{1}{2}\frac{1-c}{1+c}-r_2}{r_1-r_2},\ \ {where}\ c=2^{\frac{S(\lambda_1)-S(\lambda_2)}{r_1-r_2}},
\end{eqnarray}
where we assume $r_1\ne r_2$. (It can be easily seen from the Bloch sphere that if $r_1=r_2$, the optimal distribution would be $\{\frac{1}{2},\frac{1}{2}\}$.)

 This $\hat{p}_1$ can be used as an approximate optimal solution. Next we need numerical experiments to see how accurate $\hat{p}_1$ is.

\subsection{Numerical experiments on the approximated solution $\hat{p}_1$}

It is obvious that the maximum of Holevo quantity only depends on $r_1,r_2$ and $\theta$, so without loss of generality, we let $\rho_1$ be on the $Z$ axis and $\rho_2$ be on the $X-Z$ plane:
\begin{eqnarray}
	&\rho_1=\lambda_1|0\rangle\langle0|+(1-\lambda_1)|1\rangle\langle1|;\\
	&\rho_2 =\lambda_2|\psi_0\rangle\langle\psi_0|+(1-\lambda_2)|\psi_1\rangle\langle\psi_1|,
\end{eqnarray}
where
\begin{eqnarray}
	&|\psi_0\rangle=\cos\frac{\theta}{2}|0\rangle+\sin\frac{\theta}{2}|1\rangle;\\
	&|\psi_1\rangle=-\sin\frac{\theta}{2}|0\rangle+\cos\frac{\theta}{2}|1\rangle,
\end{eqnarray}
which means the angle between $\rho_1$ and $\rho_2$ (i.e. $\bm{r_1}$ and $\bm{r_2}$) is $\theta$. 

In the numerical experiments, we let $\lambda_1,\lambda_2$ range from $0.5$ to $1$, and $\theta$ ranges from $0$ to $\pi$.  For each value of $(\lambda_1,\lambda_2,\theta)$, we substitute $(\lambda_1,\lambda_2)$ into (\ref{app}) to compute $\hat{p}_1$.  Then substitute $(\lambda_1,\lambda_2,\theta,\hat{p}_1)$ into (\ref{chi}) to get the approximate maximum of Holevo quantity over $p_1$: $\chi(\lambda_1,\lambda_2,\theta,\hat{p}_1)$. To see how accurate this approximate maximum is, we need BA algorithm to provide an accurate maximum. The termination criterion for the iteration process of BA algorithm is, stopping when $[\max_{x}\{D(\rho_x||\rho^t)\}-f(p^t,p^t)]$ is less than $10^{-6}$, then the BA algorithm outputs a value of Holevo quantity $\chi_{BA}(\lambda_1,\lambda_2,\theta)$. And we can compute the error of $\chi(\lambda_1,\lambda_2,\theta,\hat{p}_1)$ then take the maximum over $\theta\in[0,\pi]$
\begin{eqnarray}
	{Error}(\lambda_1,\lambda_2)=\max_{\theta\in[0,\pi]}|\chi(\lambda_1,\lambda_2,\theta,\hat{p}_1)-\chi_{BA}(\lambda_1,\lambda_2,\theta)|.
\end{eqnarray}

 FIG. \ref{f1} is the numerical result, which is a plot of $(\lambda_1,\lambda_2,{Error}(\lambda_1,\lambda_2))$.





\begin{figure}[h]

\centerline{\includegraphics[width=9cm]{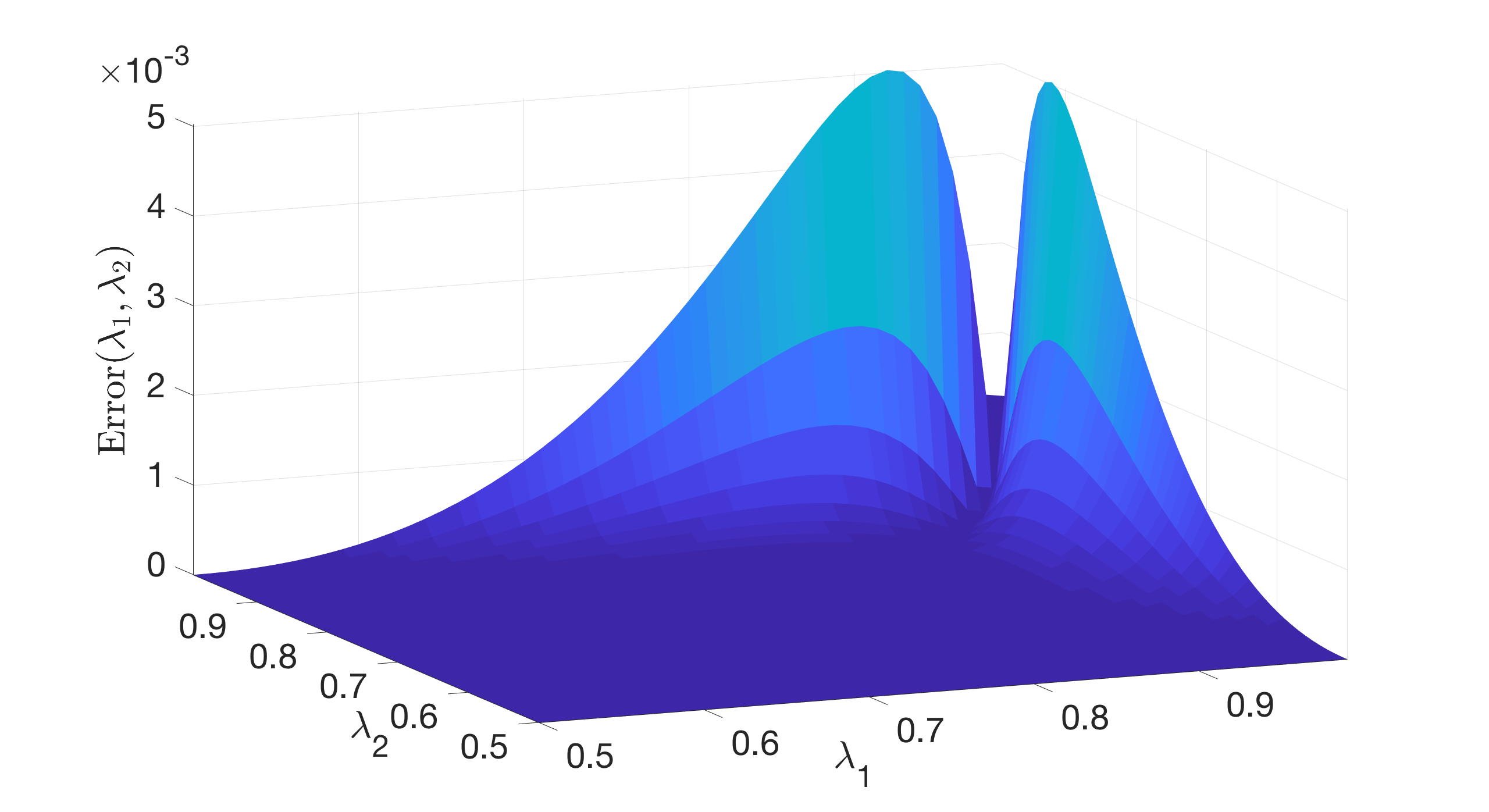}}
\caption{}	
\label{f1}
\end{figure}
From  Figure \ref{f1} we can see that if $\lambda_1,\lambda_2$ are not ``too big", the error can be upper bounded by $10^{-3}$. To see this more directly, we take the maximum of ${Error}(\lambda_1,\lambda_2)$ for different ranges of $\lambda_1,\lambda_2$: 
\begin{eqnarray}
	\max_{\lambda_1,\lambda_2\in[0.5,R]}{Error}(\lambda_1,\lambda_2).
\end{eqnarray}
Figure. \ref{f2} is a plot of $(R,\max_{\lambda_1,\lambda_2\in[0.5,R]}{Error}(\lambda_1,\lambda_2))$.
\begin{figure}[h]
\centerline{\includegraphics[width=9cm]{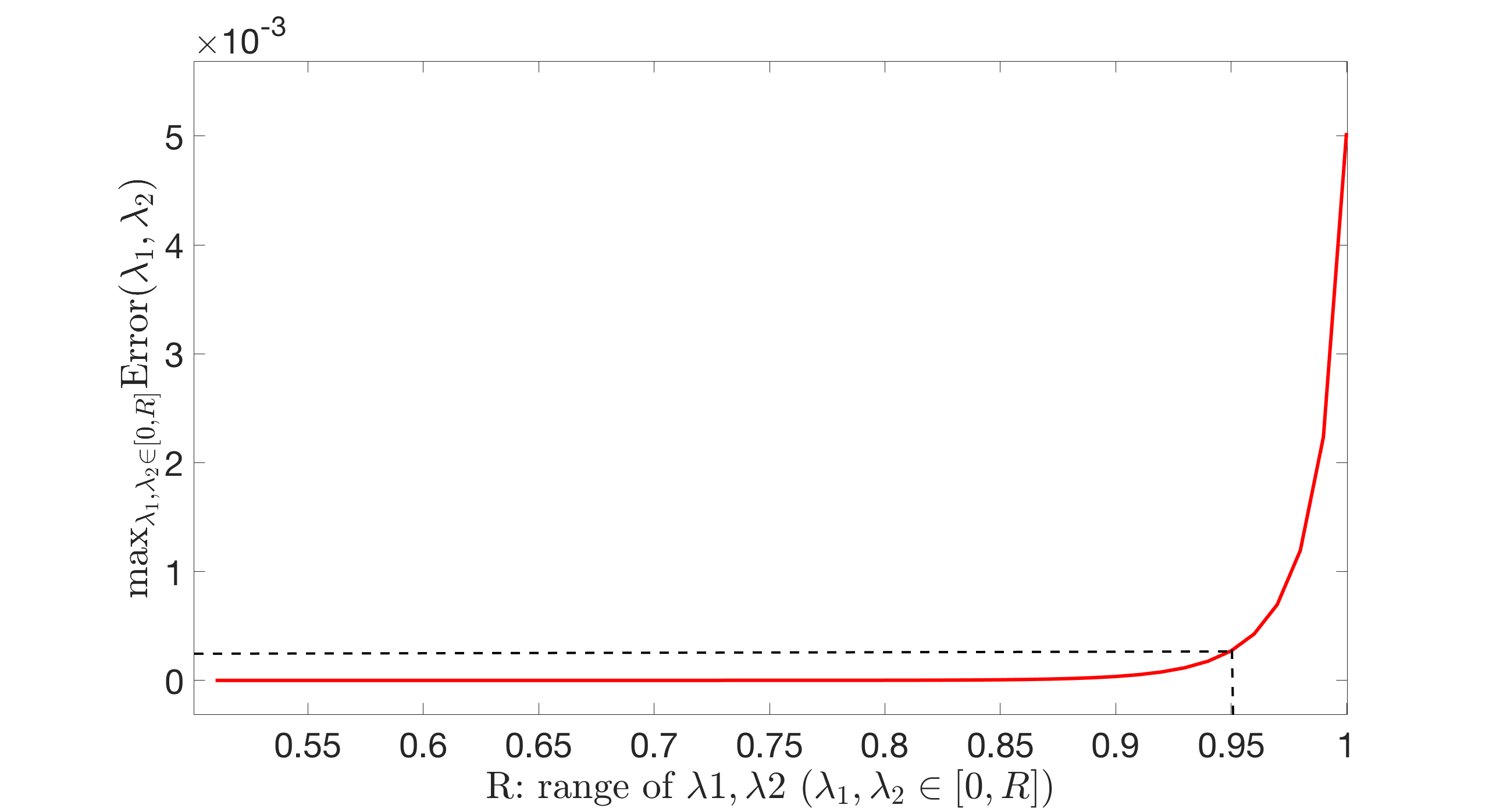}}
\caption{}	
\label{f2}
\end{figure}
From Figure. \ref{f2} we can see that if $\lambda_1,\lambda_2<0.95$, the error of approximate maximum of Holevo quantity can be upper bounded by $3\times10^{-4}$. So we can conclude that when the bigger eigenvalues of $\rho_1,\rho_2$ are not too big (no bigger than $0.95$), (\ref{app}) can make the error of the maximum of Holevo quantity smaller than $3\times10^{-4}$.

The approximate solution is an interesting phenomenon. The reason why the angle $\theta$ has such little influence on the maximum of Holevo quantity is unclear.

\section{Conclusion}
In this paper, we used numerical experiments to show that the BA algorithm works much faster than the theoretical analysis. We also presented an explicit approximate solution for the binary inputs, two dimensional outputs channel, which has an error within $3\times 10^{-4}$ when the eigenvalues of the two channel matrices are not bigger than $0.95$. However, the theoretical reason behind this approximation is not clear.
\vspace{2em}

$\bf{References}$

\vspace{1em}

\bibliography{Paper}	

\begin{thebibliography}{10}

\bibitem{myBA}
Haobo {Li} and Ning {Cai}.
\newblock {A Blahut-Arimoto Type Algorithm for Computing Classical-Quantum
  Channel Capacity}.
\newblock {\em arXiv e-prints}, page arXiv:1904.11188, Apr 2019.

\bibitem{Holevo}
A.S. Holevo.
\newblock Problems in the mathematical theory of quantum communication
  channels.
\newblock {\em Reports on Mathematical Physics}, 12(2):273 -- 278, 1977.

\bibitem{holevo1}
Alexander~S. Holevo.
\newblock The capacity of the quantum channel with general signal states.
\newblock {\em IEEE Trans- actions on Information Theory}, 44:269--273, 1998.

\bibitem{wilde}
M.~Wilde.
\newblock {\em From Classical to Quantum Shannon Theory}.
\newblock Cambridge University Press, 2016.

\bibitem{Arimoto}
S.~Arimoto.
\newblock An algorithm for computing the capacity of arbitrary discrete
  memoryless channels.
\newblock {\em IEEE Transactions on Information Theory}, 18(1):14--20, January
  1972.

\bibitem{blahut}
R.~Blahut.
\newblock Computation of channel capacity and rate-distortion functions.
\newblock {\em IEEE Transactions on Information Theory}, 18(4):460--473, 1972.

\bibitem{yeung}
Raymond~W. Yeung.
\newblock {\em Information Theory and Network Coding}.
\newblock Springer, 2008.

\bibitem{original}
H.~Nagaoka.
\newblock Algorithms of arimoto-blahut type for computing quantum channel
  capacity.
\newblock {\em Proceedings IEEE International Symposium on Information Theory
  (ISIT)}, pages 354--, 1998.

\bibitem{npc}
Salman Beigi and Peter W.~Shor.
\newblock On the complexity of computing zero-error and holevo capacity of
  quantum channels.
\newblock 10 2007.

\bibitem{Sutter}
D.~Sutter, T.~Sutter, P.~Mohajerin Esfahani, and R.~Renner.
\newblock Efficient approximation of quantum channel capacities.
\newblock {\em IEEE Transactions on Information Theory}, 62(1):578--598, Jan
  2016.

\bibitem{nestrov}
Yurii Nesterov.
\newblock Smooth minimization of non-smooth functions.
\newblock {\em Mathematical Programming}, 103:127--152, 2005.

\end{thebibliography}
\bibliographystyle{unsrt}

\end{document}